\documentclass[conference]{IEEEtran}
\usepackage{graphicx,amssymb,amsmath}
\usepackage[noadjust]{cite}%cite
\usepackage{setspace}               % set line space
\usepackage{stfloats}
\usepackage{bbding}
\usepackage{amsthm}
\usepackage{amsmath}
\usepackage{flushend}
\usepackage{cases,subeqnarray}
\usepackage{bm,multirow,bigstrut}
\usepackage{textcomp}
\usepackage{latexsym,bm}
\usepackage{booktabs,changebar}
\usepackage{xcolor}
\usepackage{mathtools}
\usepackage{dsfont}
\usepackage{extarrows}
\usepackage{mathrsfs}
\usepackage{subfigure}
\usepackage{pifont}

\usepackage{cite}
\usepackage{bm}
\usepackage{cleveref}
\usepackage{multicol}       % table
\usepackage{multirow}       % table
\usepackage{array}          % table
\usepackage{colortbl}
\usepackage{makecell}
\definecolor{crimson}{RGB}{192,0,0}         % color crimson
\definecolor{navy}{RGB}{47,85,151}         % color crimson

\makeatletter                               % algorithm
\newif\if@restonecol
\makeatother

\makeatletter
\newif\if@restonecol
\makeatother

\usepackage[linesnumbered,ruled,vlined]{algorithm2e}%[ruled,vlined]{
\usepackage{algpseudocode}
\usepackage{amsmath}
\usepackage{listings}                       % Matlab code
\usepackage[framed,numbered,autolinebreaks,useliterate]{mcode}
\theoremstyle{plain}

\newtheorem{prop}{Proposition}
\newtheorem{coro}{Corollary}

\theoremstyle{plain}
\newtheorem{rem}{Remark}

\IEEEoverridecommandlockouts

\begin{document}

%----------------------------title&author&thanks----------------------------
\title{Statistics Approximation-Enabled Distributed Beamforming for Cell-Free Massive MIMO}
\author{Zhe Wang$^\star$, Emil Bj{\"o}rnson$^\star$, Jiayi Zhang$^\dag$, Peng Zhang$^\star$, Vitaly Petrov$^\star$, Bo Ai$^\dag$\\
{\small $^\star$KTH Royal Institute of Technology and Digital Futures, Stockholm, Sweden.} {\small $^\dag$Beijing Jiaotong University, Beijing, China.}
\thanks{This work was supported in part by the research center Digital Futures in Sweden and in part by the SSF Future Research Leaders Program (Project FFL24-0100). \emph{(Corresponding author: Zhe Wang (e-mail: zhewang2@kth.se))}}}

\maketitle

%----------------------------abstract----------------------------

\begin{abstract}
We study a distributed beamforming approach for cell-free massive multiple-input multiple-output networks, referred to as Global Statistics \& Local Instantaneous information-based minimum mean-square error (GSLI-MMSE). The scenario with multi-antenna access points (APs) is considered over three different channel models: correlated Rician fading with fixed or random line-of-sight (LoS) phase-shifts, and correlated Rayleigh fading. With the aid of matrix inversion derivations, we can construct the conventional MMSE combining from the perspective of each AP, where global instantaneous information is involved. Then, for an arbitrary AP, we apply the statistics approximation methodology to approximate instantaneous terms related to other APs by channel statistics to construct the distributed combining scheme at each AP with local instantaneous information and global statistics. With the aid of uplink-downlink duality, we derive the respective GSLI-MMSE precoding schemes. Numerical results showcase that the proposed GSLI-MMSE scheme demonstrates performance comparable to the optimal centralized MMSE scheme, under the stable LoS conditions, e.g., with static users having Rician fading with a fixed LoS path.

\end{abstract}
%----------------------------keywords----------------------------
%\begin{IEEEkeywords}
%Beyond 5G network, wireless content caching, cell-free massive MIMO, submodular function.
%\end{IEEEkeywords}

%\newpage
\IEEEpeerreviewmaketitle
%\vspace{-0.2cm}
\section{Introduction}
Multiple-input multiple-output (MIMO) technology has been regarded as a vital component of wireless communication networks \cite{5595728,wang2025flexible,7827017}. Massive MIMO technology has achieved great success in fifth-generation (5G) wireless communication networks. Focusing on future wireless communication, an evolving MIMO technology, called cell-free massive MIMO (CF mMIMO), has been widely studied \cite{7827017,zhang2025performance,[162],chen2025channel}. The basic idea of CF mMIMO technology is to deploy a large number of access points (APs) over a wide coverage, which are connected to the central processing unit (CPU) via fronthaul links \cite{8000355,10522673,11049871}. By cooperating phase-coherently with the assistance of the CPU, APs can cooperatively provide uniform service quality for user equipments (UEs). Relying on this promising network topology, four main uplink processing schemes can be implemented in CF mMIMO networks \cite{[162]}, with different cooperation levels among the APs and assistance levels of the CPU. Among these uplink processing schemes, the fully centralized processing scheme and the large-scale fading decoding (LSFD) processing scheme have received considerable attention due to their excellent spectral efficiency (SE) performance.

To effectively exploit the potential capability of CF mMIMO networks, the beamforming design, including the uplink combining design and downlink precoding design, is of great importance. One promising beamforming choice is the minimum mean-square error (MMSE)-type scheme. For the uplink, as discussed in \cite{[162]}, the centralized processing scheme with centralized MMSE (C-MMSE) combining can achieve the best SE performance among various feasible processing schemes. However, the C-MMSE combining is constructed at the CPU to minimize the global decoding mean-square error (MSE), utilizing the global instantaneous channel state information (CSI). A local MMSE (L-MMSE) combining scheme can also be derived based on only local CSI at each AP \cite{[162]}, and it achieves good SE performance when used in combination with LSFD. As for the downlink, by utilizing the channel reciprocity between uplink and downlink channels, the respective MMSE-type precoding schemes can also be derived, which have been validated to achieve excellent downlink SE performance \cite{9064545}.

Although C-MMSE beamforming can achieve excellent SEs, an instantaneous CSI-based matrix inversion is required for each coherence block, which is computationally demanding when the network is large. The L-MMSE beamforming embraces lower computational complexity than the C-MMSE beamforming, but with the sacrifice of performance loss. To balance the achievable performance and computational complexity, one potential idea is to substitute some instantaneous CSI-based terms with channel statistics-based terms, which vary much slowly than the instantaneous CSI. Following this concept, some works have studied low-complexity beamforming design \cite{polegre2021pilot,OBETrans}. The authors in \cite{polegre2021pilot} considered CF mMIMO networks with single-antenna APs and UEs over Rayleigh channels and studied two beamforming schemes. The first one was reduced-complexity MMSE (RC-MMSE), which was implemented by approximating some instantaneous terms with large-scale fading coefficients using statistics approximation. The second one was generalized maximum ratio (GMR), which was constructed by replacing the matrix inversion term in the C-MMSE formulation with the statistics-based matrix after the optimization. Furthermore, motivated by the GMR scheme in \cite{polegre2021pilot}, the authors in \cite{OBETrans} investigated three optimal bilinear equalizer (OBE) combining schemes for CF mMIMO networks with multi-antenna APs over the spatially correlated Rician channel. It can be observed from \cite{OBETrans} that these statistics-driven schemes can achieve high SEs, especially in scenarios without phase-shifts or perfectly mitigated phase-shifts in the line-of-sight (LoS) term. This is because there are many global statistical cross-terms among different APs or UEs for the scenario without phase-shifts, which can be effectively utilized by OBE schemes to enhance the SE. However, the involvement of random phase-shifts makes many cross-terms equal to zero. 

Motivated by these observations, following the potential statistics-substituted idea, it would be insightful to explore some other MMSE-type schemes. These new schemes are positioned between the C-MMSE and L-MMSE schemes to balance the performance-complexity trade-off. One possible new distributed MMSE-type scheme can be formulated through the C-MMSE scheme with the aid of statistics approximation to approximate some global instantaneous terms by global statistics. In this paper, we derive a statistics approximation-enabled distributed beamforming scheme, called Global Statistics \& Local Instantaneous information-based MMSE (GSLI-MMSE), for CF mMIMO networks with multi-antenna APs and analyze it for three different channel models. The major contributions are as follows. First, we utilize the statistics approximation methodology to approximate global instantaneous CSI-based terms in the C-MMSE combining expression by global statistics on a per-AP basis. By doing so, we derive the distributed GSLI-MMSE combining scheme, which can be implemented based on local instantaneous CSI at each AP and some limited channel statistics information obtained from other APs. Then, we derive the GSLI-MMSE uplink combining and downlink precoding schemes for scenarios with multi-antenna APs over three different channel models varying from the features of line-of-sight (LoS) components: 1) correlated Rician fading with random phase-shifts or 2) correlated Rician fading without random phase-shifts in the LoS path, and 3) correlated Rayleigh fading.

\addtolength{\topmargin}{0.01in}
% \begin{itemize}
% \item We utilize the statistics approximation methodology to approximate global instantaneous CSI-based terms in the C-MMSE combining by channel statistics for each AP. By doing so, we derive the distributed GSLI-MMSE combining scheme, which can be implemented based on local instantaneous CSI at each AP and some required channel statistics information received from other APs.
% \item We derive the GSLI-MMSE combining schemes for scenarios with multi-antenna APs over three different channel models varying from the features of line-of-sight (LoS) components: correlated Rician channel with or without random LoS phase-shifts and correlated Rayleigh channel without LoS components. Indeed, based on the uplink-downlink duality, we derive the respective GSLI-MMSE precoding schemes over three channel models.
% \item Numerical results observe that these channel statistics-driven beamforming schemes perform most efficiently when LoS conditions are stable, e.g., static UEs scenario modelled by the Rician channel without random phase-shifts. This observation can bring some useful guidelines for the beamforming design over transmission scenarios with stable LoS conditions.
% \end{itemize}

\section{Fundamentals of CF mMIMO Systems}\label{sec:system}
In this paper, we study a CF mMIMO network that contains $M$ APs and $K$ UEs. The APs and UEs are arbitrarily located in a wide coverage area, where each AP and UE is equipped with $N$ antennas and a single antenna, respectively. We denote the channel response between AP $m$ and UE $k$ as $\mathbf{h}_{mk}\in \mathbb{C} ^N$, which is assumed to remain constant in each time-frequency coherence block with $\tau_c$ length. Moreover, we assume that the channel responses for different AP-UE pairs are independent and that the channel responses in different coherence blocks are independent and identically distributed. We consider the spatially correlated Rician fading channel model, where the channel response between AP $m$ and UE $k$ is modeled as
\begin{equation}\label{Rician_PS_Channel}
\mathbf{h}_{mk}=\overline{\mathbf{h}}_{mk}e^{j\theta _{mk}}+\check{\mathbf{h}}_{mk},
\end{equation}
where $\overline{\mathbf{h}}_{mk}$ is the deterministic LoS component (i.e., a scaled array response vector), $\theta _{mk}\sim \mathcal{U} [ -\pi ,\pi ]$ represents the arbitrary phase-shift between AP $m$ and UE $k$, $\check{\mathbf{h}}_{mk}\sim \mathcal{N} _{\mathbb{C}}\left( \mathbf{0},\mathbf{R}_{mk} \right)$ is the non-LoS small-scale component with $\mathbf{R}_{mk}\in \mathbb{C} ^{N\times N}$ being the spatial correlation matrix. %Moreover, we denote $\beta _{mk}^{\mathrm{NLoS}}={\mathrm{tr}( \mathbf{R}_{mk} )}/{N}$ as the NLoS large-scale fading coefficient between AP $m$ and UE $k$. 
If the UE moves, then the phase-shifts $\theta _{mk}$ vary over time at the same pace as the small-scale fading. For the scenarios with the low UE mobility or perfect estimation of these phase-shifts, \eqref{Rician_PS_Channel} can reduce to another Rician fading channel model as
$
\mathbf{h}_{mk}=\overline{\mathbf{h}}_{mk}+\check{\mathbf{h}}_{mk}.
$
Furthermore, in urban scenarios, there may not exist a LoS path between each AP-UE pair. This leads to the spatially correlated Rayleigh fading channel model with
$
\mathbf{h}_{mk}=\check{\mathbf{h}}_{mk},
$
by letting $\overline{\mathbf{h}}_{mk}=\bf{0}$. Since the channel model in \eqref{Rician_PS_Channel} is the most general one, we will use the same symbol to represent all channel models, but particularize our results when suitable. 

%\vspace{-0.2cm}
\subsection{Uplink Transmission}
%\vspace{-0.2cm}
In the uplink channel estimation phase, we utilize $\tau _p$ mutually orthogonal pilot signals to estimate the channels: $\boldsymbol{\phi} _1,\dots ,\boldsymbol{\phi} _{\tau _p}$  with $\left\| \boldsymbol{\phi} _t \right\| ^2=\tau _p$. We define $\boldsymbol{\phi} _{t_{k}}$ as the pilot signal transmitted by UE $k$. If the phase $\theta_{mk}$ can be tracked, based on the fundamentals in \cite{OBETrans}, we can compute the phase-shifts-aware MMSE estimate of $\mathbf{h}_{mk}$ as
\begin{equation}\label{CE}
\widehat{\mathbf{h}}_{mk}=\overline{\mathbf{h}}_{mk}e^{j\theta _{mk}}+\sqrt{p_k}\mathbf{R}_{mk}\mathbf{\Psi }_{mk}^{-1}\left( \mathbf{y}_{mk}^{p}-\overline{\mathbf{y}}_{mk}^{p} \right), 
\end{equation}
where $\mathbf{\Psi }_{mk}={\mathbb{E} \{ ( \mathbf{y}_{mk}^{p}-\overline{\mathbf{y}}_{mk}^{p} ) \left( \mathbf{y}_{mk}^{p}-\overline{\mathbf{y}}_{mk}^{p} \right) ^H \}}/{\tau _p}=\sum_{l\in \mathcal{P} _k}{p_l\tau _p\mathbf{R}_{ml}}+\sigma ^2\mathbf{I}_N$, $\mathbf{y}_{mk}^{p}=\sum_{l\in \mathcal{P} _k}{\sqrt{p_l}\tau _p\mathbf{h}_{ml}}+\mathbf{n}_{mk}^{p}$, and $\overline{\mathbf{y}}_{mk}^{p}=\sum_{l\in \mathcal{P} _k}{\sqrt{p_l}\tau _p\overline{\mathbf{h}}_{ml}}e^{j\theta _{ml}}$, $\mathbf{n}_{mk}^{p}=\mathbf{N}_{m}^{p}\boldsymbol{\phi} _{k}^{*}\sim \mathcal{N} _{\mathbb{C}}( \mathbf{0},\tau _p\sigma ^2\mathbf{I}_N ) $. Here, $\mathcal{P} _k$ denotes the subset of UEs that apply the same pilot signal $\boldsymbol{\phi} _{t_{k}}$ as UE $k$, $p_k$ denotes the transmit power of UE $k$, and $\sigma ^2$ is the noise power. The channel estimation error is defined as $\tilde{\mathbf{h}}_{mk}=\mathbf{h}_{mk}-\widehat{\mathbf{h}}_{mk} $, which is conditionally independent from the channel estimate $\widehat{\mathbf{h}}_{mk}$. Moreover, we have $\mathbb{E} \{ \widehat{\mathbf{h}}_{mk} |\theta _{mk} \} =\overline{\mathbf{h}}_{mk}e^{j\theta _{mk}}$, $\mathrm{Cov}\{ \widehat{\mathbf{h}}_{mk} |\theta _{mk} \} =\widehat{\mathbf{R}}_{mk}=p_k\tau _p\mathbf{R}_{mk}\mathbf{\Psi }_{mk}^{-1}\mathbf{R}_{mk}$, $\mathbb{E} \{ \tilde{\mathbf{h}}_{mk} \} =\mathbf{0}$, and $\mathrm{Cov} \{ \tilde{\mathbf{h}}_{mk} \} =\mathbf{C}_{mk}=\mathbf{R}_{mk}-p_k\tau _p\mathbf{R}_{mk}\mathbf{\Psi }_{mk}^{-1}\mathbf{R}_{mk}$. In each coherence blocks, $\tau_u$ symbols are used for uplink data transmission. All UEs transmit simultaneously and the received data signal at AP $m$ is
$
\mathbf{y}_m=\sum_{k=1}^K{\mathbf{h}_{mk}x_k}+\mathbf{n}_m\in \mathbb{C} ^N
$
with $\mathbf{n}_m\sim \mathcal{N} _{\mathbb{C}}( \mathbf{0},\sigma ^2\mathbf{I}_N ) $ and $x_k\sim \mathcal{N} _{\mathbb{C}}( 0,p_k) $ being the additive noise at AP $m$ and the uplink data symbol sent by UE $k$, respectively.

We consider both the centralized and distributed receive combining methods and will evaluate the resulting achievable uplink SE performance. For the centralized processing scheme, the decoding phase is implemented at the CPU, while all APs serve only as relays. More specifically, all pilot signals and uplink data signals received at all APs are transmitted to the CPU via the fronthauls. Thus, the received uplink data signal at the CPU can be constructed as 
$
\mathbf{y}=\sum_{k=1}^K{\mathbf{h}_kx_k}+\mathbf{n},
$
where $\mathbf{y}=[ \mathbf{y}_{1k}^{T},\dots ,\mathbf{y}_{Mk}^{T} ] ^T\in \mathbb{C} ^{MN}$ is the collective received signal and $\mathbf{h}_k=\left[ \mathbf{h}_{1k}^{T},\dots ,\mathbf{h}_{Mk}^{T} \right] ^T=\overline{\mathbf{h}}_{k}^{\prime}+\check{\mathbf{h}}_k\in \mathbb{C} ^{MN}$ represents the collective channel for UE $k$ with $\overline{\mathbf{h}}_{k}^{\prime}=[ \overline{\mathbf{h}}_{1k}^{T}e^{j\theta _{1k}},\dots ,\overline{\mathbf{h}}_{Mk}^{T}e^{j\theta _{Mk}} ] ^T$ and $\check{\mathbf{h}}_k=[ \check{\mathbf{h}}_{1k}^{T},\dots ,\check{\mathbf{h}}_{Mk}^{T} ] ^T$. The CPU can then implement centralized channel estimation for $\mathbf{h}_k$ based on the collective pilot-related signal $\mathbf{y}_{k}^{p}=[ \mathbf{y}_{1k}^{p,T},\dots ,\mathbf{y}_{Mk}^{p,T} ] ^T\in \mathbb{C} ^{MN}$. More specifically, the estimate of $\mathbf{h}_k$ can be computed as
\begin{equation}\label{CE_CPU}
\widehat{\mathbf{h}}_k=\overline{\mathbf{h}}_{k}^{\prime}+\sqrt{p_k}\mathbf{R}_k\mathbf{\Psi }_{k}^{-1}( \mathbf{y}_{k}^{p}-\overline{\mathbf{y}}_{k}^{p} ),
\end{equation}
where $\mathbf{R}_k=\mathbb{E} \{ \check{\mathbf{h}}_k\check{\mathbf{h}}_{k}^{H} \} =\mathrm{diag}( \mathbf{R}_{1k},\dots ,\mathbf{R}_{Mk} ) \in \mathbb{C} ^{MN\times MN}$ is block-diagonal since the channels are independent for different AP-UE pairs, $\overline{\mathbf{y}}_{k}^{p}=[ \overline{\mathbf{y}}_{1k}^{p,T},\dots ,\overline{\mathbf{y}}_{Mk}^{p,T} ] \in \mathbb{C} ^{MN}$, and $\mathbf{\Psi }_k=\mathrm{diag}( \mathbf{\Psi }_{1k},\dots ,\mathbf{\Psi }_{Mk} ) \in \mathbb{C} ^{MN\times MN}$. Based on \eqref{CE_CPU}, the CPU can select an arbitrary combining scheme $\mathbf{v}_k\in \mathbb{C} ^{MN}$ to derive the decoding estimate of $x_k$ with $k=1,\dots ,K $, which can be denoted as
$
\widehat{x}_k=\mathbf{v}_{k}^{H}\mathbf{y}=\mathbf{v}_{k}^{H}\mathbf{h}_kx_k+\sum_{l\ne k}^K{\mathbf{v}_{k}^{H}\mathbf{g}_lx_l}+\mathbf{v}_{k}^{H}\mathbf{n}.
$
By using the use-and-then-forget (UatF) bound \cite[Ch.~4]{8187178}, we can quantify the achievable uplink SE for UE $K$ as
$
\mathrm{SE}_{k}^{\mathrm{c},\mathrm{ul}}=\frac{\tau_u}{\tau_c}\log _2\left( 1+\mathrm{SINR}_{k}^{\mathrm{c},\mathrm{ul}} \right), 
$
where $\mathrm{SINR}_{k}^{\mathrm{c},\mathrm{ul}}$ is the effective signal-to-interference-plus-noise ratio (SINR) for UE $k$, given as
\begin{equation}\label{SINR_centrlized_UatF}
\begin{aligned}
&\mathrm{SINR}_{k}^{\mathrm{c},\mathrm{ul}}=\\
&\frac{p_k| \mathbb{E} \{ \mathbf{v}_{k}^{H}\mathbf{h}_k \} |^2}{\sum\limits_{l=1}^K{p_l\mathbb{E} \{ | \mathbf{v}_{k}^{H}\mathbf{h}_l |^2 \}}-p_k| \mathbb{E} \{ \mathbf{v}_{k}^{H}\mathbf{h}_k \} |^2+\sigma ^2\mathbb{E} \{ \| \mathbf{v}_k \| ^2 \}},
\end{aligned}
\end{equation}
where the expectations consider all sources of randomness. %The proof of this capacity bound can be found in \cite[Appendix C.3.4]{8187178}.

With centralized receiver processing, the desirable  combining scheme is the C-MMSE combining, defined as
\begin{equation}\label{CMMSE}
\mathbf{v}_k=p_k\left( \sum_{l=1}^K{p_l\left( \widehat{\mathbf{h}}_l\widehat{\mathbf{h}}_{l}^{H}+\mathbf{C}_l \right)}+\sigma ^2\mathbf{I}_{MN} \right) ^{-1}\widehat{\mathbf{h}}_k,
\end{equation}
where $\mathbf{C}_{l}={\mathbf{R}}_{l}-p_l\tau _p{\mathbf{R}}_{l}\mathbf{\Psi }_{l}^{-1}{\mathbf{R}}_{l}=\mathrm{diag}( \mathbf{C}_{1l},\dots ,\mathbf{C }_{Ml} )\in \mathbb{C} ^{MN\times MN}$ is the covariance matrix of $\tilde{\mathbf{h}}_l=\mathbf{h}_l-\widehat{\mathbf{h}}_l$. The C-MMSE combining scheme in \eqref{CMMSE} minimizes the conditional mean-squared error $\mathrm{MSE}_k=\mathbb{E} \{ |  x_k-\mathbf{v}_{k}^{H}\mathbf{y} |^2 |\{ \widehat{\mathbf{h}}_k ,\theta_k\} \} $ with $\theta_k$ being the collection of phase-shifts for UE $k$.\footnote{C-MMSE also maximizes the achievable SE when using the instantaneous information-based capacity bound in \cite[Eq. (11)]{[162]}.}

If distributed receiver processing is used, each AP can compute its channel estimates in \eqref{CE} locally and perform local soft decoding. More specifically, based on the local channel estimates $\widehat{\mathbf{h}}_{mk}$, AP $m$ can design a local combining scheme $\mathbf{v}_{mk}$ for UE $k$ with $k=1,\dots ,K $. Based on the received data signal $\mathbf{y}_m$, the local estimate of $x_k$ at AP $m$ can be computed as
$
\widehat{x}_{mk}=\mathbf{v}_{mk}^{H}\mathbf{y}_m\!=\!\mathbf{v}_{mk}^{H}\mathbf{h}_{mk}x_k\!\!+\!\!\sum_{l\ne k}^K{\mathbf{v}_{mk}^{H}\mathbf{h}_{ml}x_l}\!+\!\mathbf{v}_{mk}^{H}\mathbf{n}_m.
$
One well-studied local combining scheme is the local MMSE combining (L-MMSE) given by
\begin{equation}\label{local_MMSE}
\begin{aligned}
\mathbf{v}_{mk}=p_k\left( \sum_{l=1}^K{p_l\left( \widehat{\mathbf{h}}_{ml}\widehat{\mathbf{h}}_{ml}^{H}+\mathbf{C}_{ml} \right)}+\sigma ^2\mathbf{I}_N \right) ^{-1}\widehat{\mathbf{h}}_{mk},
\end{aligned}
\end{equation}
which minimizes the local mean-squared error $\mathrm{MSE}_{mk}=\mathbb{E} \{ |  x_k-\mathbf{v}_{mk}^{H}\mathbf{y}_m |^2 |\{ \widehat{\mathbf{h}}_{mk},\theta_{mk}\} \} $. Then, all APs transmit the local decoding data to the CPU, where the CPU can compute a weighted estimate of $x_k$ based on the LSFD coefficients $\{ a_{mk}:m=1,\dots ,M \} $ as 
$
\check{x}_k=\sum_{m=1}^M{a_{mk}^{*}\widehat{x}_{mk}}=( \sum_{m=1}^M{a_{mk}^{*}\mathbf{v}_{mk}^{H}\mathbf{h}_{mk}} ) x_k+\sum_{l\ne k}^K{\sum_{m=1}^M{a_{mk}^{*}\mathbf{v}_{mk}^{H}\mathbf{h}_{ml}x_l}}+\mathbf{n}_{k}^{\prime}
$
with $\mathbf{n}_{k}^{\prime}=\sum_{m=1}^M{a_{mk}^{*}\mathbf{v}_{mk}^{H}\mathbf{n}_m}$. Based on $\check{x}_k$, we can compute the SE for UE $k$ as 
$
\mathrm{SE}_{k}^{\mathrm{d,ul}}=\frac{\tau _d}{\tau _c}\log _2(1+\mathrm{SINR}_{k}^{\mathrm{d,ul}}),
$
with 
\begin{equation}\label{SINR_LSFD}
\mathrm{SINR}_{k}^{\mathrm{d,ul}}\!\!=\!\!\frac{p_k\left| \mathbf{a}_{k}^{H}\mathbb{E} \left\{ \mathbf{d}_{kk} \right\} \right|^2}{\mathbf{a}_{k}^{H}\!( \sum\limits_{l=1}^K{p_l\mathbf{\Theta }_{kl}}\!-\!p_k\mathbb{E} \left\{ \mathbf{d}_{kk} \right\} \mathbb{E} \left\{ \mathbf{d}_{kk}^{H} \right\} +\sigma ^2\mathbf{D}_k ) \mathbf{a}_k}
\end{equation}
being the effective SINR for UE $k$ in this case, where $\mathbf{d}_{kk}=[ \mathbf{v}_{1k}^{H}\mathbf{h}_{1k},\dots ,\mathbf{v}_{Mk}^{H}\mathbf{h}_{Mk} ] ^T\in \mathbb{C} ^M$, $\mathbf{a}_k=[ a_{1k},\dots ,a_{Mk} ] \in \mathbb{C} ^M$, and $\mathbf{D}_k=\mathrm{diag}[ \mathbb{E} \{ \| \mathbf{v}_{1k} \| ^2 \} ,\dots ,\mathbb{E} \{ \| \mathbf{v}_{Mk} \| ^2 \} ] \in \mathbb{C} ^{M\times M}$. Moreover, we have $\mathbf{\Theta }_{kl}\in \mathbb{C} ^{M\times M}$, where the $(n,m)$-th element of $\mathbf{\Theta }_{kl}$ is  $[ \mathbf{\Theta }_{kl} ] _{nm}=\mathbb{E} \{( \mathbf{v}_{mk}^{H}\mathbf{h}_{ml} ) ^H( \mathbf{v}_{nk}^{H}\mathbf{h}_{nl} ) \}$. It is worth noting that \eqref{SINR_LSFD} is maximized by 
$
\mathbf{a}_{k}^{*}\!\!=\!\!( \sum_{l=1}^K{p_l\mathbf{\Theta }_{kl}}-p_k\mathbb{E} \left\{ \mathbf{d}_{kk} \right\} \mathbb{E} \left\{ \mathbf{d}_{kk}^{H} \right\} +\sigma ^2\mathbf{D}_k ) ^{-1}\mathbb{E} \left\{ \mathbf{d}_{kk} \right\}.
$
leading to $\mathrm{SINR}_{k}^{\mathrm{d,ul},*}=p_k\mathbb{E} \left\{ \mathbf{d}_{kk}^{H} \right\} \mathbf{a}_{k}^{*}$.

\subsection{Downlink Transmission}

In the downlink transmission phase, $\tau_d$ data symbols are transmitted. We assume the channel reciprocity between the uplink and downlink channels. We denote the downlink transmitted signal from AP $m$ by 
$
\mathbf{x}_m=\sum_{k=1}^K{\mathbf{g}_{mk}\varsigma _k}\in \mathbb{C} ^{N}.
$
with $\mathbf{g}_{mk}\in \mathbb{C} ^{N}$ being the downlink precoding vector for UE $k$ at AP $m$ and $\varsigma _k\sim \mathcal{N} _{\mathbb{C}}( 0,1 )$ being the downlink data symbol for UE $k$. We can further represent $\mathbf{g}_{mk}$ as $\mathbf{g}_{mk}=\mu _{mk}\overline{\mathbf{g}}_{mk}$, where $\overline{\mathbf{g}}_{mk}$ denotes an arbitarily scaled precoding vector for UE $k$ at AP $m$ and $\mu _{mk}=\sqrt{{p_{mk}}/{\mathbb{E} \{ \| \overline{\mathbf{g}}_{mk} \| ^2 \}}}$ denotes the power factor between AP $m$ and UE $k$ with 
$p_{mk}$ being the allocated power for UE $k$ at AP $m$. The transmit power is constrained by $\sum_{k=1}^K{\mathbb{E} \{ \| \mathbf{g}_{mk} \| ^2 \}}\leqslant p_m$, where $p_m$ denotes the total downlink transmit power by AP $m$. Thanks to the channel reciprocity and resulting uplink-downlink duality, the downlink precoding scheme can be designed based on the uplink combining scheme \cite{8187178}. The received signal for UE $k$ can be formulated as
\begin{equation}
\begin{aligned}
y_k^{\mathrm{dl}}&=\sum_{m=1}^M{\mathbf{h}_{mk}^{H}\mathbf{x}_m}+n_k\\
&=\sum_{m=1}^M{\mathbf{h}_{mk}^{H}\mathbf{g}_{mk}}\varsigma _k+\sum_{l\ne k}^K{\sum_{m=1}^M{\mathbf{h}_{mk}^{H}\mathbf{g}_{ml}}\varsigma _l}+n_k,
\end{aligned}
\end{equation}
where $n_k\sim \mathcal{N} _{\mathbb{C}}( 0,\sigma ^2) $ denotes the downlink noise at UE $k$ with noise power $\sigma ^2$. Then, the achievable downlink SE for UE $k$ can be denoted as 
$
\mathrm{SE}_{k}^{\mathrm{dl}}=\frac{\tau _d}{\tau _c}\log _2( 1+\mathrm{SINR}_{k}^{\mathrm{dl}}),
$
where the downlink SINR is computed based on the UatF bound as
\begin{equation}\label{SINR_DL}
\mathrm{SINR}_{k}^{\mathrm{dl}}\!\!=\!\!\frac{ \left| \sum\limits_{m=1}^M{\mathbb{E} \{ \mathbf{g}_{mk}^{H}\mathbf{h}_{mk} \}} \right|^2}{\sum\limits_{l=1}^K{\mathbb{E} \left\{ \left| \sum\limits_{m=1}^M{\mathbf{g}_{ml}^{H}\mathbf{h}_{mk}} \right|^2 \! \right\} -} \left| \!\sum\limits_{m=1}^M{\mathbb{E} \{ \mathbf{g}_{mk}^{H}\mathbf{h}_{mk} \}} \right|^2\!\!+\sigma ^2}.
\end{equation}

\section{Distributed GSLI-MMSE Beamforming}

The results in the previous section are well-established in the CF mMIMO literature. However, even if the C-MMSE combining in \eqref{CMMSE} is optimal in terms of achievable SE, the computation of the C-MMSE combining vectors necessitates global instantaneous information and leads to high computational complexity since the matrix inversion is high-dimensional. In this section, we develop a novel low-complexity distributed beamforming design that takes inspiration from C-MMSE.
To this end, we note that \eqref{CMMSE} can be reformulated as
\begin{equation}\label{CMMSE_Construct}
\mathbf{v}_k=p_k\left( \widehat{\mathbf{H}}\mathbf{P}\widehat{\mathbf{H}}^H+\mathbf{W} \right) ^{-1}\widehat{\mathbf{H}}\mathbf{e}_k,
\end{equation}
where $\widehat{\mathbf{H}}\triangleq [ \widehat{\mathbf{h}}_1,\dots ,\widehat{\mathbf{h}}_K ] \in \mathbb{C} ^{MN\times K}$, $\mathbf{P}=\mathrm{diag}( p_1,\dots ,p_K ) \in \mathbb{C} ^{K\times K}$, 
\begin{equation}
\mathbf{W}=\sum_{l=1}^K{p_l \mathbf{C}_l }+\sigma ^2\mathbf{I}_{MN}\in \mathbb{C} ^{MN\times MN},
\end{equation}
and $\mathbf{e}_k$ denotes the $k$-th column of $\mathbf{I}_K$.
Then, based on the matrix inversion lemma \cite[Lemma B.3]{8187178}
\begin{equation}\label{matrix_inversion}
( \mathbf{D}+\mathbf{EFE}^H ) ^{-1}\mathbf{E}=\mathbf{D}^{-1}\mathbf{E}( \mathbf{E}^H\mathbf{D}^{-1}\mathbf{E}+\mathbf{F}^{-1} ) ^{-1}\mathbf{F}^{-1},
\end{equation}
with $\mathbf{D}\in \mathbb{C} ^{N_1\times N_1}$, $\mathbf{E}\in \mathbb{C} ^{N_1\times N_2}$, and $\mathbf{F}\in \mathbb{C} ^{N_1\times N_2}$,
we can further reformulate \eqref{CMMSE_Construct} as
\begin{equation}\label{CMMSE_Construct_inverse}
\mathbf{v}_k=p_k\mathbf{W}^{-1}\widehat{\mathbf{H}}\left( \widehat{\mathbf{H}}^H\mathbf{W}^{-1}\widehat{\mathbf{H}}+\mathbf{P}^{-1} \right) ^{-1}\mathbf{P}^{-1}\mathbf{e}_k,
\end{equation}
by letting the matrices in \eqref{CMMSE_Construct} be $\mathbf{D}=\mathbf{W}$, $\mathbf{E}=\widehat{\mathbf{H}}$, and $\mathbf{F}=\mathbf{P}$, respectively. %Note that the matrix inversion result \eqref{matrix_inversion} can be easily derived based on the standard matrix inversion method in \cite[Lemma B.3]{8187178}. 
Note that, since $\mathbf{C}_l$ and $\mathbf{I}_{MN}$ are block-diagonal matrices, $\mathbf{W}$ is also a block-diagonal matrix, which can be constructed as 
\begin{equation}
\mathbf{W}=\mathrm{diag}( \mathbf{W}_1,\dots ,\mathbf{W}_M ),
\end{equation}
where $\mathbf{W}_m=\sum_{l=1}^K{p_l\mathbf{C}_{ml}}+\sigma ^2\mathbf{I}_N\in \mathbb{C} ^{N\times N}$. By now utilizing the fact that 
$\mathbf{v}_k=[ \mathbf{v}_{1k}^{T},\dots ,\mathbf{v}_{Mk}^{T} ] ^T$ and relying on the block-diagonal feature of $\mathbf{W}$, we can represent $\mathbf{v}_{mk}$ as
\begin{equation}\label{vmk}
\mathbf{v}_{mk}=p_k\mathbf{W}_{m}^{-1}\widehat{\mathbf{H}}_m\left( \widehat{\mathbf{H}}^H\mathbf{W}^{-1}\widehat{\mathbf{H}}+\mathbf{P}^{-1} \right) ^{-1}\mathbf{e}_k,
\end{equation}
where $\widehat{\mathbf{H}}_m=[ \widehat{\mathbf{h}}_{m1},\dots ,\widehat{\mathbf{h}}_{mK} ] \in \mathbb{C} ^{N\times K}$ represents the local channel estimate matrix at AP $m$. 

We notice that each AP can compute its vector $\mathbf{v}_{mk}$  in a distributed manner using \eqref{vmk}, but it requires knowledge of the global instantaneous channel matrix $\widehat{\mathbf{H}}$, which must be sent to each AP $m$ via the fronthaul, leading to a significant signaling transmission burden. We can further formulate \eqref{vmk} as
\begin{equation}\label{vmk_derivation}
\begin{aligned}
\mathbf{v}_{mk}=p_k\mathbf{W}_{m}^{-1}\widehat{\mathbf{H}}_m \left( \widehat{\mathbf{H}}_{m}^{H}\mathbf{W}_{m}^{-1}\widehat{\mathbf{H}}_m\!\!+\!\!\!\sum_{m^{\prime}\ne m}^M{\mathbf{\Xi }_{m^{\prime}}}\!+\!\mathbf{P}^{-1} \right) ^{-1}\mathbf{e}_k
\end{aligned}
\end{equation}
where $\mathbf{\Xi }_{m^{\prime}}=\widehat{\mathbf{H}}_{m^{\prime}}^{H}\mathbf{W}_{m^{\prime}}^{-1}\widehat{\mathbf{H}}_{m^{\prime}}$. Note that \eqref{vmk_derivation} can be easily derived by using the notation 
\begin{equation}
\widehat{\mathbf{H}}\triangleq  \left[ \widehat{\mathbf{H}}_1;\dots ;\widehat{\mathbf{H}}_m;\dots ;\widehat{\mathbf{H}}_M \right] \in \mathbb{C} ^{MN\times K}
\end{equation} 
and relying on the block-diagonal feature of $\mathbf{W}$. As observed from \eqref{vmk_derivation}, for AP $m$, the required instantaneous terms from the other APs in \eqref{vmk} can be denoted by $\sum_{m^{\prime}\ne m}^M{\mathbf{\Xi }_{m^{\prime}}}$. 

\subsection{Proposed Method}

Instead of passing around the term $\sum_{m^{\prime}\ne m}^M{\mathbf{\Xi }_{m^{\prime}}}$ between the APs in each coherence block, we propose a new combining scheme by approximating this term by its average. Note that the statistics remain constant for each realization of APs/UEs locations, so it can be easily shared over the fronthaul. Our goal is to achieve SEs close to that with C-MMSE combining, while using the same amount of instantaneous information per AP as with L-MMSE combining.
We call our method \emph{Global Statistics \& Local Instantaneous information-based MMSE (GSLI-MMSE)} combining, and define it as
\begin{equation} \label{vmk_derivation2}
\mathbf{v}_{mk}=p_k\mathbf{W}_{m}^{-1}\widehat{\mathbf{H}}_m \! \!\left( \!\widehat{\mathbf{H}}_{m}^{H}\mathbf{W}_{m}^{-1}\widehat{\mathbf{H}}_m\!+\!\!\!\sum_{m^{\prime}\ne m}^M\!\!\!{\mathbb{E} \left\{ \mathbf{\Xi }_{m^{\prime}} \right\} \!}+\mathbf{P}^{-1} \!\!\right)^{\!\!-1} \!\!\!\mathbf{e}_k.
\end{equation}
The expectation can be computed as follows.

\begin{prop}\label{prop_GSLI}
The average of $\mathbf{\Xi }_{m^{\prime}} \in \mathbb{C} ^{K\times K}$ in \eqref{vmk_derivation2} is called $\overline{\mathbf{\Xi }}_{m^{\prime}}\in \mathbb{C} ^{K\times K}$.
For the Rician fading channel model in \eqref{Rician_PS_Channel},
the $(k,l)$-th element of $\overline{\mathbf{\Xi }}_{m^{\prime}}$ is given by 
\begin{equation}
\begin{aligned}\label{approximate_Rician_ps}
&\left[ \overline{\mathbf{\Xi }}_{m^{\prime}} \right] _{kl}=\\ 
&\begin{cases}
	\begin{array}{l}
	0,\,\,\,\mathrm{if}\,\,l\notin \mathcal{P} _k\,\,\\
	\!\!\!\!\!\sqrt{p_kp_l}\tau _p\mathrm{tr}\left( \mathbf{R}_{m^{\prime} l}\mathbf{\Psi }_{m^{\prime} k}^{-1}\mathbf{R}_{m^{\prime} k}\mathbf{W}_{m^{\prime}}^{-1} \right) ,\,\,\,\mathrm{if}\,\,l\in \mathcal{P} _k\backslash \left\{ k \right\}\\
\end{array}\\
	\mathrm{tr}\!\left( \mathbf{W}_{m^{\prime}}^{-1}\overline{\mathbf{H}}_{m^{\prime} k} \right) +\!\!p_k\tau _p\mathrm{tr}\left( \mathbf{R}_{m^{\prime} k}\mathbf{\Psi }_{m^{\prime} k}^{-1}\mathbf{R}_{m^{\prime} k}\mathbf{W}_{m^{\prime}}^{-1} \right) \!,\mathrm{if}\,l=k\\
\end{cases}
\end{aligned}
\end{equation}
where $\overline{\mathbf{H}}_{m^{\prime} kk}=\overline{\mathbf{h}}_{m^{\prime} k}\overline{\mathbf{h}}_{m^{\prime} k}^{H}\in \mathbb{C} ^{N\times N}$.
\end{prop}
\begin{IEEEproof}
The goal is to compute $\mathbb{E} \{ \mathbf{\Xi }_{m^{\prime}} \}$. Note that the $(k,l)$-th element of $\mathbf{\Xi }_{m^{\prime}}$ can be denoted as 
\begin{equation}
[ \mathbf{\Xi }_{m^{\prime}} ] _{kl}=\widehat{\mathbf{h}}_{m^{\prime} k}^{H}\mathbf{W}_{m^{\prime}}^{-1}\widehat{\mathbf{h}}_{m^{\prime} l}.
\end{equation}
Thus, we can compute its expectation by considering the following three cases. For $l\notin \mathcal{P} _k$, we have 
$\mathbb{E} \{\widehat{\mathbf{h}}_{m^{\prime} k}^{H}\mathbf{W}_{m^{\prime}}^{-1}\widehat{\mathbf{h}}_{m^{\prime} l}\}=\mathbb{E} \{( \overline{\mathbf{h}}_{m^{\prime} k}^{ H}e^{-j\theta _{m^{\prime} k}}+\widehat{\mathbf{h}}_{m^{\prime} k,\left( 1 \right)}^{H} ) \mathbf{W}_{m^{\prime}}^{-1}( \overline{\mathbf{h}}_{m^{\prime} l}e^{j\theta _{m^{\prime} l}}+\widehat{\mathbf{h}}_{m^{\prime} l,\left( 1 \right)} ) \}=0$ 
since $\widehat{\mathbf{h}}_{m^{\prime} k,\left( 1 \right)} $ and $\widehat{\mathbf{h}}_{m^{\prime} l,\left( 1 \right)} $ are mutually independent with $\widehat{\mathbf{h}}_{m^{\prime} k,\left( 1 \right)}=\sqrt{p_k}\mathbf{R}_{m^{\prime} k}\mathbf{\Psi }_{m^{\prime} k}^{-1}(\mathbf{y}_{m^{\prime} k}^{p}-\bar{\mathbf{y}}_{m^{\prime} k}^{p})$. Since the phase-shifts in the LoS components vary between coherence blocks and are independent for different AP-UE pairs. Thus, all LoS-related expectation terms are zero. For $l\in \mathcal{P} _k\backslash \left\{ k \right\}$, $\widehat{\mathbf{h}}_{m^{\prime} k,\left( 1 \right)}$ and $\widehat{\mathbf{h}}_{m^{\prime} l,\left( 1 \right)}$ are correlated, we have 
\begin{equation}
\mathbb{E} \{ \widehat{\mathbf{h}}_{m^{\prime} l,\left( 1 \right)}\widehat{\mathbf{h}}_{m^{\prime} k,\left( 1 \right)}^{H} \} =\sqrt{p_kp_l}\tau _p\mathrm{tr}\left( \mathbf{R}_{m^{\prime}l}\mathbf{\Psi }_{m^{\prime}k}^{-1}\mathbf{R}_{m^{\prime}k}\mathbf{W}_{m^{\prime}}^{-1} \right).
\end{equation}
The other LoS-related terms are still $0$ due to random phase-shifts. Finally, for $l=k$, we have additional term $\mathrm{tr}( \mathbf{W}_{m^{\prime}}^{-1}\overline{\mathbf{H}}_{m^{\prime} k} )$ compared to the case $l\in \mathcal{P} _k\backslash \left\{ k \right\}$, which can be easily derived. We hence have obtained the result in \eqref{approximate_Rician_ps}.
\end{IEEEproof}

The expectations of many LoS-related terms are $0$ in Proposition~\ref{prop_GSLI} due to the existence of the random phase-shifts $\theta _{mk}$ that appear when the UE moves.
In scenarios where the fading is caused by moving multipath objects or the LoS phase-shifts can be tracked, \eqref{Rician_PS_Channel} reduces to $\mathbf{h}_{mk}=\overline{\mathbf{h}}_{mk}+\check{\mathbf{h}}_{mk}$. 
We then obtain the following proposition.

\begin{prop}\label{GSLI_wo_PS}
For spatially correlated Rician channels without phase-shifts, the average of $\mathbf{\Xi }_{m^{\prime}}$ in \eqref{vmk_derivation2} is called $\overline{\mathbf{\Xi }}_{m^{\prime} ,\left( 1 \right)}$ and the $(k,l)$-th element is given by
\begin{equation}
\begin{aligned}\label{term_wo_ps}
[\overline{\mathbf{\Xi }}_{m^{\prime},\left( 1 \right)}]_{kl}&=\mathrm{tr}\!\left( \mathbf{W}_{m^{\prime}}^{-1}\overline{\mathbf{H}}_{m^{\prime}lk} \right) \\
&+\begin{cases}
	0,\,\,\,\mathrm{if}\,\,l\notin \mathcal{P} _k\,\,\\
	\!\!\sqrt{p_kp_l}\tau _p\mathrm{tr}\!\left( \mathbf{R}_{m^{\prime}l}\mathbf{\Psi }_{m^{\prime}k}^{-1}\mathbf{R}_{m^{\prime}k}\mathbf{W}_{m^{\prime}}^{-1} \right) \!,\mathrm{if}\,l\!\in \!\mathcal{P} _k\\
\end{cases}
\end{aligned}
\end{equation}
where $\overline{\mathbf{H}}_{m^{\prime} lk}=\overline{\mathbf{h}}_{m^{\prime} l}\overline{\mathbf{h}}_{m^{\prime} k}^{H}\in \mathbb{C} ^{N\times N}$.
\end{prop}
\begin{IEEEproof}
We have $\mathbf{h}_{mk}=\overline{\mathbf{h}}_{mk}+\check{\mathbf{h}}_{mk}$ with $\overline{\mathbf{h}}_{mk}$ being deterministic. Thus, the expectations of LoS-related terms in Proposition~\ref{prop_GSLI} become non-zero. We can compute \eqref{term_wo_ps} following similar steps as in the proof of Proposition~\ref{prop_GSLI}.
\end{IEEEproof}

\begin{rem}
The block-fading model normally rely on that coherence blocks are independent across both time and frequency, but the phase-shifts in the LoS path is constant over the frequency domain. Hence, for systems operating over a narrow time interval and using a wide bandwidth, it is feasible to perfectly estimate phase-shifts and compensate for them. In these situations, one should use model and results from Proposition~\ref{GSLI_wo_PS} instead of those in Proposition~\ref{prop_GSLI}.
\end{rem}

When the LoS path is zero, we obtain a Rayleigh fading channel model leading to the following special case.

\begin{coro}\label{GSLI_Rayleigh}
For correlated Rayleigh fading channels, the average of $\mathbf{\Xi }_{m^{\prime}}$ in \eqref{vmk_derivation2} is called $\overline{\mathbf{\Xi }}_{m^{\prime} ,\left( 2 \right)}$ and
the $(k,l)$-th element is $\sqrt{p_kp_l}\tau _p\mathrm{tr}( \mathbf{R}_{m^{\prime} l}\mathbf{\Psi }_{m^{\prime} k}^{-1}\mathbf{R}_{m^{\prime} k}\mathbf{W}_{m^{\prime}}^{-1} ) $ for $l\!\in \!\mathcal{P} _k$ and $0$ for $l\notin \mathcal{P} _k$, respectively.
\end{coro}

The main difference between these results is how the average of $\mathbf{\Xi }_{m^{\prime}}$ depends on the LoS-related components. 

Based on the uplink GSLI-MMSE combining schemes, we can derive the corresponding downlink GSLI-MMSE precoding schemes by plugging the GSLI-MMSE combining vector in \eqref{vmk_derivation2} into $\overline{\mathbf{g}}_{mk}=\mathbf{v}_{mk}$, and computing the average depending on the channel model.

\section{Numerical Results}
In the numerical results, we investigate a CF mMIMO network with $1\times 1 \, \mathrm{km}^2$ coverage area. All APs and UEs are randomly located in the coverage area. The large-scale fading coefficients are modeled following the same method as in \cite{OBETrans}. For the Rician channel model, we assume that there exist LoS paths between all AP-UE pairs. The spatial channel correlation matrix $\mathbf{R}_{mk} $ is constructed based on the Gaussian local scattering model defined in \cite[Sec. 2.6]{8187178}. Furthermore, we have $\tau _c=200$, $\tau _p=1$, $p_k=200 \, \mathrm{mW}$, $\sigma ^2=-94 \, \mathrm{dBm}$, and the bandwidth is $20 \, \mathrm{MHz}$. Each coherence block is either used for only uplink or downlink data transmission, that is $\tau _u=\tau _d=\tau _c-\tau _p$. For the downlink, we have $p_m=K\times200\, \mathrm{mW}$. For other parameters, we refer to \cite{OBETrans} due to space constraints. Note that the parameter settings used in this section are illustrative. The proposed formulations are general and are applicable to arbitrary numbers of APs, antennas per AP, and UEs, as well as to arbitrary AP/UE placements and coverage regions, etc.

\begin{figure}[t]
\centering
\includegraphics[width=0.9\columnwidth]{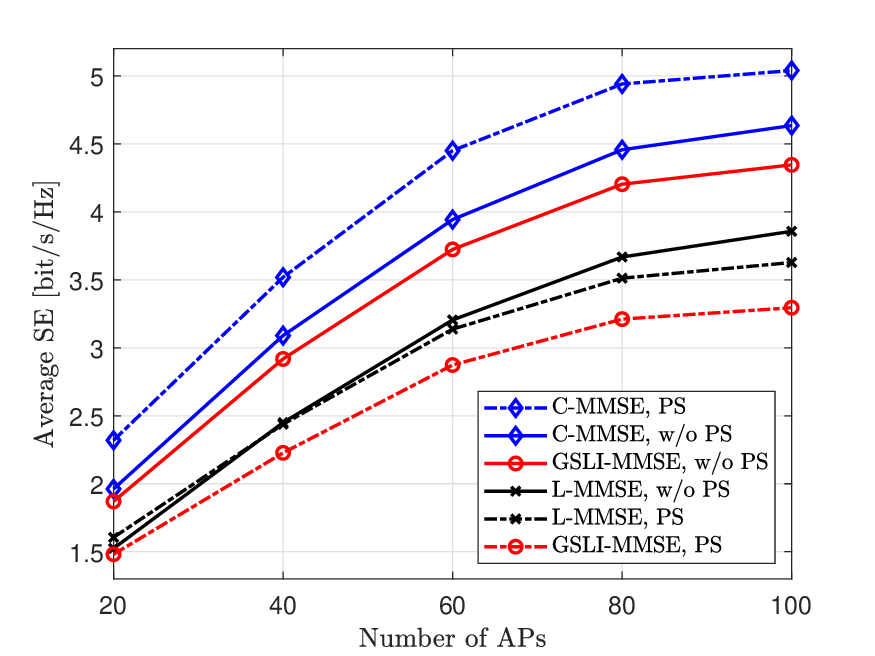}
\vspace{-0.2cm}
\caption{Average uplink SE against $M$ over Rician fading channel models with or without phase-shifts with $K=40$ and $N=4$. ``PS" and ``w/o PS" denote the Rician channel models with phase-shifts and without phase-shifts, respectively. ``C-MMSE" denotes the centralized processing scheme with the C-MMSE combining. ``GSLI-MMSE" and ``L-MMSE" denote the distributed processing schemes with optimal LSFD strategies over the GSLI-MMSE and L-MMSE combining schemes, respectively. \label{1}}
\vspace{-0.5cm}
\end{figure}

\begin{figure}[t]
\centering
\includegraphics[width=0.9\columnwidth]{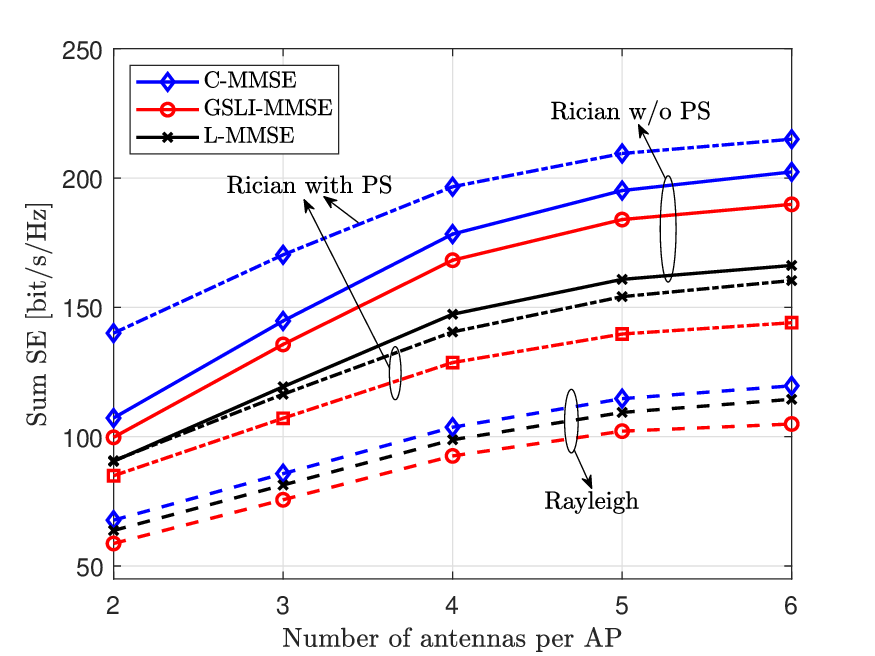}
\vspace{-0.2cm}
\caption{Sum uplink SE against $N$ for different combining schemes over different channel models with $M=80$ and $K=40$. \label{2}}
\vspace{-0.3cm}
\end{figure}

\begin{figure}[t]
\centering
\includegraphics[width=0.9\columnwidth]{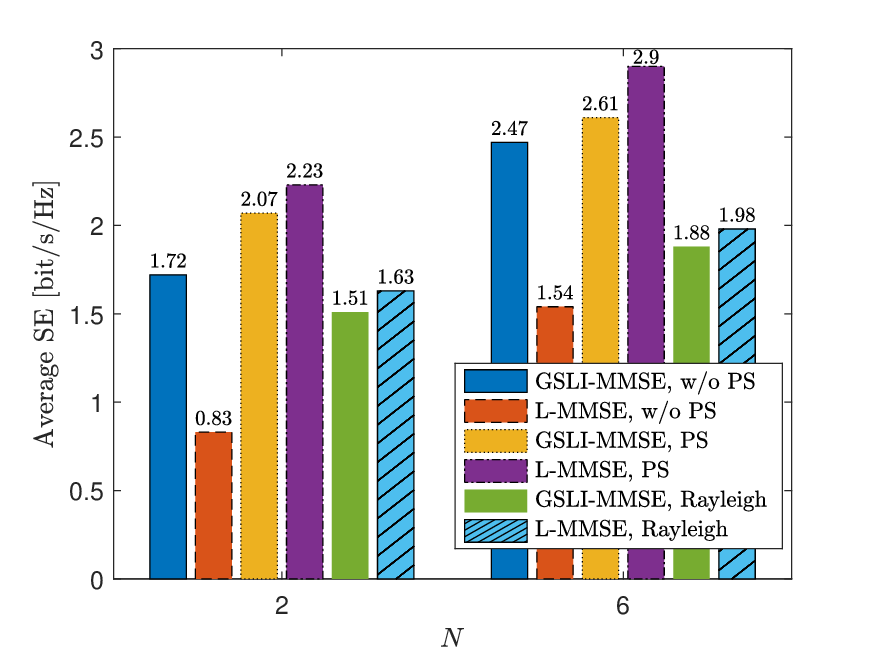}
\vspace{-0.2cm}
\caption{Average downlink SE for different precoding schemes over different channel models with $M=80$ and $K=40$. \label{3}}
\vspace{-0.3cm}
\end{figure}

% \begin{figure*}[htbp]
% 	\centering
% 	\begin{minipage}{0.4\linewidth}
% 		\centering
%   \includegraphics[scale=0.5]{Fig_UL_Rician_N_1013.eps}\vspace{-0.3cm}
% \caption{Sum uplink SE against $N$ for different combining schemes and channel models with $M=80$ and $K=40$. \label{2}}
% 	\end{minipage}
%     	\begin{minipage}{0.4\linewidth}
% 		\centering
%   \includegraphics[scale=0.5]{Fig_DL_Pillar_1007.eps}\vspace{-0.4cm}
% \caption{Average downlink SE for different precoding schemes over different channel models with $M=80$ and $K=40$. \label{3}}
% 	\end{minipage}
%   \vspace{-0.7cm}
% \end{figure*}

Fig.~\ref{1} showcases the average uplink SE against the number of APs, $M$, for different processing schemes over Rician fading channel models with/without phase-shifts. For both Rician channel models, the C-MMSE combining achieves the best SE performance due to its global processing capability. Over the Rician fading channel model with phase-shifts, the L-MMSE combining outperforms the GSLI-MMSE combining. However, for the Rician channel model without phase-shifts, the GSLI-MMSE combining can outperform the L-MMSE combining. These observations are because the GSLI-MMSE combining over the Rician channel model without phase-shifts can efficiently exploit many channel statistics-based cross terms among different APs or UEs, as formalized in Proposition~\ref{GSLI_wo_PS}. However, over the Rician channel model with phase-shifts, due to the existence of random phase-shifts, many expectations of LoS-related cross terms become $0$, so only very limited statistical information can be utilized by the GSLI-MMSE combining, as shown in Proposition~\ref{prop_GSLI}. These arguments suggest that the channel statistics approximation-driven combining schemes, such as the GSLI-MMSE combining, are most effective over the transmitting scenario with stable LoS components, such as the Rician model without phase-shifts.

To further demonstrate the argument above, we study the sum uplink SE versus the number of antennas per AP $N$ over different channel models in Fig.~\ref{2}. As observed, the GSLI-MMSE combining scheme outperforms the L-MMSE combining scheme only over the Rician channel model without phase shifts, e.g., achieving about a $14\%$ SE improvement by the GSLI-MMSE combining compared to the L-MMSE combining for $N=4$. However, over the Rician channel model with phase-shifts and the Rayleigh channel model without LoS components, L-MMSE combining outperforms the proposed GSLI-MMSE combining. The performance gaps between the GSLIM-MMSE and L-MMSE combining schemes over the Rician channel model with phase-shifts and the Rayleigh channel model are about $8\%$. These observations further demonstrate the excellent effectiveness of the channel statistics approximation-driven combining schemes under the scenario with stable LoS components.

Fig.~\ref{3} studies the average downlink SE performance for the GSLI-MMSE and L-MMSE precoding schemes over different channel models. We can observe that a similar argument as above can be derived that the GSLI-MMSE precoding can only outperform the L-MMSE precoding over the Rician channel model without phase-shifts. For instance, the GSLI-MMSE precoding can achieve $60\%$ SE improvement compared to the L-MMSE precoding for $N=6$. Meanwhile, the performance gaps between the GSLI-MMSE and L-MMSE precoding schemes over the Rician channel model with phase-shifts and the Rayleigh channel model are $10\%$ and $5\%$, respectively, which also aligns with the argument above.

Building on the above observations and our prior findings in \cite{OBETrans}, we highlight several guidelines for beamforming in CF mMIMO networks. Under transmission scenarios with stable LoS components, e.g., static UE scenario modeled by the Rician channel without random phase-shifts, utilizing channel statistics-driven beamforming at each AP can enhance the SE performance over other distributed beamforming schemes by substituting certain instantaneous CSI-based terms in MMSE-type schemes with channel statistics-based terms. To implement this, only some global channel statistics components, which remain constant over a long period of time, are required to be transmitted to each AP. Schemes of this kind, including the proposed GSLI-MMSE and the optimal bilinear equalizer (OBE) in \cite{OBETrans}, are effective when LoS conditions are stable, where interference directions are reliably identifiable. By contrast, when LoS conditions are unstable (Rician model with random phase-shifts) or absent (Rayleigh model), one needs to utilize instantaneous CSI related to the most strongly interfering devices to achieve good interference suppression.

\section{Conclusion}\label{sec:conclusion}
We proposed the distributed GSLI-MMSE beamforming scheme for CF mMIMO networks with multi-antenna APs and analyzed it for three different Rician or Rayleigh channel models. For the uplink, for each AP, we replaced the global instantaneous terms from the other APs in the classical C-MMSE combining formulations with their respective statistical averages and derived the GSLI-MMSE combining schemes in closed form. Then, based on uplink–downlink duality, we extended the GSLI-MMSE combining schemes to respective downlink precoding schemes. Numerical results show that in the scenarios with stable LoS components (Rician fading without random phase-shifts), the proposed GSLI-MMSE beamforming notably and consistently outperforms L-MMSE beamforming in both uplink SE and downlink SE across a wide range of scenarios. In such scenarios, the proposed GSLI-MMSE scheme demonstrates performance comparable to the optimal C-MMSE solution, while featuring distributed implementation at each AP.

%\begin{appendices}

%\section{Proof of Proposition \ref{coro:random SINR}}\label{appe:1}

%\newcounter{mytempeqncnt}
%\begin{figure*}[b]
%\normalsize
%\setcounter{mytempeqncnt}{\value{equation}}
%\hrulefill
%\vspace*{4pt}
%\setcounter{equation}{15}
%\begin{align}\label{eq:er_inid_x}\notag
%$aaa\\
%$bbb
%\end{align}
%\setcounter{equation}{\value{mytempeqncnt}}
%\end{figure*}
%\end{appendices}

%\pagebreak[4]

\bibliographystyle{IEEEtran}
\bibliography{IEEEabrv,Ref}

\end{document}